\documentclass[runningheads]{llncs}

\usepackage{booktabs}
\usepackage{amssymb}
\usepackage{amsmath}
\usepackage{graphicx}
\usepackage[font=small,labelfont=bf]{caption}
\usepackage[subrefformat=parens,labelfont=default]{subcaption}
\usepackage{mathtools}
\mathtoolsset{showonlyrefs=true}
\usepackage[disable]{todonotes}
\usepackage{hyperref}
\usepackage[inline]{enumitem}
\usepackage{multirow}
\usepackage{wrapfig}
\usepackage{dcolumn}
\usepackage{tabularx}
\usepackage{cite}
\usepackage{timetravel}
\usepackage[misc]{ifsym}

\newcommand{\lncsarxiv}[2]{#2}

\let\doendproof\endproof
\def\endproof{\hfill $\qed$\doendproof}

\DeclareMathOperator{\bc}{bc}
\DeclareMathOperator{\g}{g}
\newcommand{\bco}{{\ensuremath{\bc^\circ}}}

\newcommand{\fedges}{F}
\newcommand{\fvertices}{V_F}
\newcommand{\surfgk}{\mathcal{S}}
\newcommand{\surfsub}{\Omega}

\graphicspath{{figures/}}

\newcommand{\review}[1]{}

\newtheorem{observation}{Observation}

\begin{document}

\title{Bundled Crossings Revisited\thanks{\lncsarxiv{The
      full version of this article is available at
      ArXiv~\cite{full-version-arxiv}. }{}M.K.\ was supported by
    DAAD; S.C.\ was supported by DFG grant WO$\,$758/11-1.}}

\author{Steven Chaplick\inst{1}\lncsarxiv{$^{\textrm{(\Letter)}}$\orcidID{0000-0003-3501-4608}}{}
\and Thomas~C.~van~Dijk\inst{1}\lncsarxiv{\orcidID{0000-0001-6553-7317}}{}
\and Myroslav Kryven\inst{1}
\and Ji-won Park\inst{2}
\and Alexander~Ravsky\inst{3}
\and Alexander Wolff\inst{1}\lncsarxiv{\orcidID{0000-0001-5872-718X}}{}
}

\authorrunning{S.~Chaplick et al.}

\institute{
Universit\"at W\"urzburg, W\"urzburg,
Germany\\ \email{\lncsarxiv{steven.chaplick@uni-wuerzburg.de}{firstname.lastname@uni-wuerzburg.de}}
\and
KAIST, Daejeon, Republic of
Korea \lncsarxiv{}{\\\email{wldnjs1727@kaist.ac.kr}}
\and
Pidstryhach Institute for Applied Problems of Mechanics
and Mathematics, \\
National Academy of Sciences of Ukraine, Lviv,
Ukraine \lncsarxiv{}{\\\email{alexander.ravsky@uni-wuerzburg.de}}
}
\maketitle              %

\setcounter{footnote}{0} %

\begin{abstract}
  An effective way to reduce clutter in
  a graph drawing that has (many) crossings is to group
  edges that travel in parallel into \emph{bundles}.
  Each edge can participate in many such bundles.
  Any crossing in this bundled
  graph occurs between two bundles, i.e., as a \emph{bundled crossing}.
  We consider the problem of bundled crossing minimization:
  A graph is given and the goal is to find
  a bundled drawing with at most $k$ bundled crossings.
  We show that the problem is NP-hard when we require a simple drawing.
  Our main result is an FPT algorithm (in $k$)
  when we require a simple circular layout.
  These results make use of the connection between bundled crossings
  and graph genus.
\end{abstract}

\section{Introduction}

\review{R1:
Figure captions are not handled consistently, sometimes subfigures
are indicated by (a), (b), etc (e.g. Figure 6), sometimes only by
semicolons (e.g. Figure 4), but in other figures semicolons have a
different function (e.g. Figure 2). Can you please be consistent?}
In traditional node--link diagrams,
vertices are mapped to points in
the plane and edges are usually drawn as straight-line segments
connecting the vertices.  For large and dense graphs,
however, such layouts tend to be so cluttered that it is hard to see
any structure in the data.  For this reason, Holten~\cite{h-hebva-TVCG06}
introduced \emph{bundled drawings}, where edges that are close
together and roughly go into the same direction are drawn using
B\'ezier curves such that the grouping becomes visible.  Due to the
practical effectiveness of this approach, it has quickly been adopted
by the InfoVis community
\cite{Cui_2009,pupyrev11,Gansner_2011,Hurter_2012,Hurter_2014}.
However, bundled drawings have only recently attracted study from a
theoretical point of view
\cite{afp-bcn-GD16, fhsv-bceg-LATIN16, fpw-omlbbc-JGAA15,
dfflmrsw-bcsv-JGAA17}.

Crossing minimization is a fundamental problem in graph
drawing~\cite{s-gcnvs-EJC17}.
Its natural generalization in bundled drawings
is bundled crossing minimization,
see Definition~\ref{def:bundledcrossing} for the formalization of a
bundled crossing.
In his survey on crossing
minimization, Schaefer lists the bundled
crossing number as a variant of the crossing number and suggests to
study it \cite[page~35]{s-gcnvs-EJC17}.

\paragraph{Related Work.}
Fink et al.~\cite{fpw-omlbbc-JGAA15} considered bundled crossings
(which they called block crossings) in the context of drawing metro maps.
A metro network is a
planar graph where vertices are stations and metro lines are simple paths
in this graph. These paths representing metro lines can share edges.
They enter an edge at one endpoint
in some linear order, follow the edge as x-monotone curves
(considering the edge as horizontal),
and then leave the edge at the other endpoint in some linear order.
In order to improve the readability of metro maps, the authors
suggested to bundle crossings.
The authors then studied the problem of minimizing bundled crossings
in such metro maps.
Fink et al.\ also introduced \emph{monotone}
bundled crossing minimization where each pair of lines can intersect
at most once.
Later, Fink et al.~\cite{dfflmrsw-bcsv-JGAA17} applied the concept
of bundled crossings to drawing storyline visualizations.
A storyline visualization is a set of x-monotone curves where the
x-axis represents time in a story.  Given a set of \emph{meetings}
(subsets of the curves that must be consecutive at given
points in time), the task is to find a drawing that realizes
the meetings and minimizes the number of bundled crossings.
Fink et al.\ showed that, in this setting, minimizing bundled
crossings is fixed-parameter tractable (FPT)
and can be approximated in a restricted case.
Our research builds on recent works of Fink et
al.~\cite{fhsv-bceg-LATIN16} and Alam et al.~\cite{afp-bcn-GD16}, who
extended the notion of bundled crossings from sets of x-monotone curves
to general drawings of graphs~-- details below.

\paragraph{Notation and Definitions.}

In graph drawing, it is common to define a drawing of a graph as a function
that maps vertices to points in the plane and edges to Jordan arcs
that connect the corresponding points.  In this paper, we are less restrictive
in that we sometimes allow edges to self-intersect.  We will often identify
vertices with their points and edges with their curves.
Moreover, we assume that each pair of edges shares at most a finite
number of points, that edges can touch (that is, be tangent to) each
other only at endpoints,
\review{R3: It is unclear what you mean by touching: intersection, tangency?
I guess you mean that at every intersection between two the open
Jordan arcs the two arcs intersect transversely.}
and that any point of the plane that is not a vertex is contained in
at most two edges.
A drawing is \emph{simple} if any two edges intersect at
most once and no edge self-intersects.  We consider both simple and
non-simple drawings; look ahead at Fig.~\ref{fig:k33} for a simple and a
non-simple drawing of~$K_{3,3}$.

\begin{definition}[Bundled Crossing]\label{def:bundledcrossing}
Let $D$ be a drawing, not necessarily simple, and let $I(D)$ be the set
of intersection points among the edges
\review{R3: Are vertices included?}
(not including the vertices)
in~$D$.  We say that a
\emph{bundling} of~$D$ is a partition of~$I(D)$ into \emph{bundled
  crossings}, where a set $B \subseteq I(D)$ is a bundled crossing if
the following holds (see Fig.~\ref{fig:abc}).
\begin{itemize}[noitemsep,topsep=0pt]
\item $B$ is contained in a closed Jordan region~$R(B)$ whose boundary
  \review{R3: In real analysis, a Jordan region is a set whose boundary
    has 0 measure.  I'm not sure this is what you really mean. You
    certainly don't mean a closed region bounded by a Jordan curve,
    since you include degenerate regions in l.87.}
  consists of four Jordan arcs~$\tilde e_1$, $\tilde e_2$, $\tilde
  e_3$, and~$\tilde e_4$ that are pieces of edges~$e_1$, $e_2$, $e_3$,
  and~$e_4$ in~$D$ (a piece of an edge $e$ is $D[e]\bigcap R(B)$);
  when the edge pieces are not distinct, we define $R(B)$ not
  as a Jordan region but as an arc or a point.
\item The pieces of the edges cut out by the region~$R(B)$ can be
  partitioned into two sets $\tilde E_1$ and $\tilde E_2$ such that
  $\tilde e_1, \tilde e_3 \in \tilde E_1$, $\tilde e_2, \tilde e_4 \in
  \tilde E_2$, and each pair of edge pieces in $\tilde E_1 \times
  \tilde E_2$ has exactly one intersection point in~$R(B)$, whereas no
  two edge pieces in $\tilde E_1$ (respectively $\tilde{E_2}$) have a
  common point in $R(B)$.
\end{itemize}
\end{definition}

\begin{figure}[tb]
  \includegraphics[page=1]{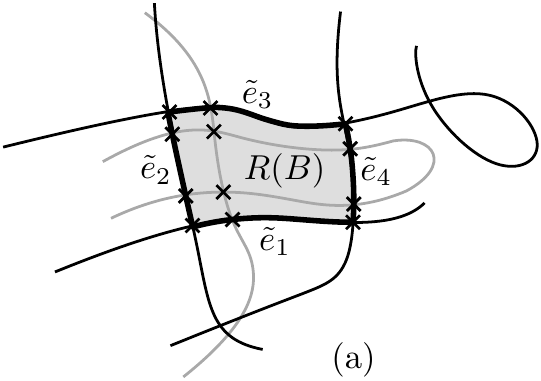} \hfil \includegraphics[page=2]{abc}
  \caption{(a)~A non-degenerate bundled crossing~$B$ and (b)~a
    degenerate bundled crossing~$B'$;
    crossings belonging to a bundled crossing are marked
    with crosses}
  \label{fig:abc}
\end{figure}
Our definition is similar to that of Alam et al.~\cite{afp-bcn-GD16}
but defines the Jordan region $R(B)$ more precisely.
We call the sets $\tilde E_1$ and $\tilde E_2$ of edge pieces
\emph{bundles} and the Jordan arcs $\tilde e_1, \tilde e_3 \in \tilde
E_1$ and $\tilde e_2, \tilde e_4 \in \tilde E_2$ \emph{frame arcs} of
the bundles $\tilde E_1$ and $\tilde E_2$, respectively.
For simple drawings, we accordingly call the edges that bound the two
bundles of a bundled crossing \emph{frame edges}.
We say that a bundled crossing is \emph{degenerate} if at least one of
the bundles consists of only one edge piece; see
Fig.~\ref{fig:abc}(b).  In this case, the region of the plane
associated with the crossing coincides with that edge piece.
In particular, any point in $I(D)$ by itself is a degenerate bundled
crossing.
Hence, every drawing admits
a trivial bundling.

We use $\bc(G)$ to denote the \emph{bundled crossing number} of a
graph~$G$, i.e., the smallest number of bundled crossings over all
bundlings of all simple drawings of~$G$.  When we do not insist on
simple drawings, we denote the corresponding number by~$\bc'(G)$.
In the circular setting, where vertices are required to lie on the
boundary of a disk and edges inside this disk, we consider the
analogous \emph{circular bundled crossing numbers} $\bco(G)$
and $\bco'(G)$ of a graph~$G$.

Fink et al.~\cite{fhsv-bceg-LATIN16} showed that it is NP-hard to
compute the minimum number of bundled crossings that a given drawing
of a graph can be partitioned into.
They also showed that this problem generalizes the problem of
partitioning a rectilinear polygon with holes into the minimum number
of rectangles, and they exploited this connection to construct a
10-approximation for computing the number of bundled crossings in the
case of a \emph{fixed circular drawing}.  They left open the
computational complexity of the general and the circular bundled
crossing number for the case that
the drawing is not fixed.

Alam et al.~\cite{afp-bcn-GD16} showed that $\bc'(G)$ equals the
orientable genus of~$G$, which in general is NP-hard to
compute~\cite{Thomassen_1989}.  They also showed that there is a graph~$G$
with~$\bc'(G) \ne \bc(G)$ by proving that $\bc'(K_6)=1<\bc(K_6)$.
As it turns out, the two problem variants differ in the circular
setting, too (see Fig.~\ref{fig:k33} and Observation~\ref{obs:k33}).
For computing $\bc(G)$ and $\bco(G)$, Alam et al.~\cite{afp-bcn-GD16}
gave an algorithm whose approximation factor depends on the density of
the graph.
They posed the existence of
an FPT algorithm for $\bco(G)$ as an open question.

\paragraph{Our Contribution.}

\review{R3: This sounds intriguing. It would be good to tell right
  here how they are different.  Do you simply mean that they are not
  equal?}
As some graphs~$G$ have $\bc'(G) \ne \bc(G)$ (see
Fig.~\ref{fig:k33}), %
Fink et al.~\cite{fhsv-bceg-LATIN16} posed the complexity of computing
the bundled crossing number $\bc(G)$ of a given graph~$G$ as an open
problem.  We settle this %
in Section~\ref{sec:bcnp-hard} as
follows:
\begin{theorem}
  \label{thm:bchard}
  Given a graph $G$, it is NP-hard to compute $\bc(G)$.
\end{theorem}
Our main result, which we prove in Section~\ref{sec:main}, resolves an
open question of Alam et al.~\cite{afp-bcn-GD16} concerning the
fixed-parameter tractability of bundled crossing minimization in
circular layouts as follows:
\review{R3:(and earlier): "circular layout" It is unclear why you
  introduce new name for a 1-page drawing.  AW: I would state
  somewhere early on that ``circular'' is the same as ``1-page''.}

\newcommand{\thmbcofpt}{%
\label{thm:bco-FPT} \sloppy
    There is a computable function $f$ such that, for any $n$-vertex
    graph~$G$ and integer~$k$, we can check, in $O(f(k) n)$ time,
    whether $\bco(G) \le k$, i.e., whether~$G$ admits a circular
    layout with $k$ bundled crossings.
    Within the same time bound, such a layout can be computed.}
\begin{theorem}
  \thmbcofpt
\end{theorem}
To prove this, we use an approach similar to that of Bannister and
Eppstein~\cite{Bannister_2014} for 1-page crossing minimization
(that is, edge crossing minimization in a circular layout).
Bannister and Eppstein observe that the set of crossing edges of a
circular layout with $k$ edge crossings of a graph~$G$ forms
an arrangement of curves that partition the drawing into $O(k)$
subgraphs, each of which occurs in a distinct face of this arrangement.
The subgraphs are obviously outerplanar.  This means that $G$ has
bounded treewidth (see \lncsarxiv{the full
  version~\cite{full-version-arxiv}}{Appendix~\ref{sec:mso2+courcelle}}).
So, by enumerating all ways to draw the
crossing edges of a circular layout with $k$ edge crossings, and,
for each such way, expressing the edge
partition problem (into crossing edges and
outerplanar components) in extended monadic
second order logic (MSO$_2$),
\emph{Courcelle's Theorem}~\cite{courcelle1990}
(stated as Theorem~\ref{thm:courcelle}
in Section~\ref{sec:main}) can be applied (leading to fixed-parameter
tractability).

The difficulty in using this approach for bundled crossing minimization
is in showing how to partition the graph into a set of $O(k)$
``crossing edges'' (our analogy will be the frame edges)
and a collection of $O(k)$ outerplanar graphs.
This is where we exploit the connection to genus.
Moreover, constructing an MSO$_2$ formula is somewhat more
difficult in our case due to the more complex way our regions interact
with our special set of edges.

\section{Computing $\bc(G)$ is NP-Hard}
\label{sec:bcnp-hard}

For a given graph $G$, finding a drawing
with the fewest bundled crossings
resembles
computing
the \emph{orientable genus}\footnote{I.e., computing the fewest
  \emph{handles} to attach to the sphere so that $G$ can be drawn on
  the resulting surface without any crossings.}
$\g(G)$ of $G$.  In fact, Alam et al.~\cite{afp-bcn-GD16} showed that
$\bc'(G) = \g(G)$.
Thus, deciding %
$\bc'(G)=k$ for some $k$ is
NP-hard and that it is FPT in $k$, since the same
holds for deciding %
$\g(G)=k$~\cite{Thomassen_1989, mohar1999linear, KawarabayashiMR08}.

\begin{theorem}[\!\!\cite{afp-bcn-GD16}]
  \label{obs:bc'np-hard}
 For every graph $G$ with genus $k$, it holds that $\bc'(G) = k$.
\end{theorem}
To show this, Alam et al.~\cite{afp-bcn-GD16} first showed
that a drawing with $k$ bundled crossings can be lifted onto a surface
of genus $k$, and thus $\bc'(G) \ge \g(G)$:
\begin{observation}[$\!\!$\cite{afp-bcn-GD16}]
  \label{obs:lift}
  A drawing $D$ with $k$ bundled crossings can be \emph{lifted} onto a
  surface of genus $k$ via a one-to-one correspondence between bundled
  crossings and handles, i.e., at each bundled crossing, we attach a handle
  for one of the two edge bundles, thus providing a crossing-free
  \emph{lifted drawing}; see Fig.~\ref{fig:dwarf-village}.
\end{observation}
Then, to see that
$\bc'(G) \le \g(G)$, Alam et al.~\cite{afp-bcn-GD16} used the \emph{fundamental polygon}
representation (or \emph{polygonal schema})
\cite{d-ctgs-HDCG17} of a drawing on a genus-$g$ surface.
More precisely, the sides of the polygon are
numbered  in circular order $a_1,b_1,a_1',b_1', \ldots,
a_g,b_g,a_g',b_g'$; for $1 \le k \le g$, the pairs $(a_k,a_k')$ and
$(b_k,b_k')$ of sides are identified in opposite direction, meaning
that an edge leaving side $a_k$ appears on the corresponding
position of side $a_k'$; see Fig.~\ref{fig:k6-tourus}
for an example
showing $K_6$ drawn in a fundamental square, which models a drawing on
the torus. In such a representation, all vertices lie in the interior
of the fundamental polygon and
all edges leave the polygon avoiding vertices of the polygon.
Alam et al.~\cite{afp-bcn-GD16} showed that such a
representation can be transformed into a non-simple bundled drawing with
$g$ many bundled crossings.
It is not clear, however, when such a representation can be transformed into
a simple bundled drawing with $g$ bundled crossings, as this transformation
can produce drawings with self-loops and pairs of edges
crossing multiple times, e.g., Alam et al.~\cite[Lemma 1]{afp-bcn-GD16}
showed that $\bc(K_6) = 2$ while $\bc'(K_6) = \g(K_6) = 1$.

We show that the problem remains NP-hard for simple drawings.
\begin{proof} [of Theorem~\ref{thm:bchard}]
    Let $G'$ be the graph obtained from $G$ by subdividing each edge
    $O(|E(G)|^2)$ times.  We reduce from the NP-hardness of computing
    the genus $\g(G)$ of~$G$ by showing that $\bc(G') = \g(G)$, with
    Observation~\ref{obs:lift} in mind.

    Consider the embedding of $G$ onto the genus-$\g(G)$ surface.  By
    a result of Lazarus et al.~\cite[Theorem~1]{lpvv-ccpso-SoCG01}, we
    can construct a fundamental polygon representation
    of the embedding so that its boundary
    intersects with edges of the graph $O(\g(G)|E(G)|)$ times.  Note
    that each edge piece outside the polygon intersects each other
    edge piece at most once; see Fig.~\ref{fig:k6-tourus}.  We
    then subdivide the edges by adding a vertex to each intersection
    \review{R3: Do you subdivide the edge (as suggested by the next sentence)?}
    of an edge with the boundary
    of the fundamental polygon.  This subdividing of edges ensures that no
    edge intersects itself or intersects another edge more than once
    in the corresponding drawing of the graph on the plane; hence, the
    drawing is simple.  Since $\g(G) \le |E(G)|$, by
    subdividing edges further whenever necessary, we obtain a drawing
    of~$G'$.  Our
    subdivisions keep the integrity of all bundled crossings, so
    $\bc(G') \le \g(G)$.  On the other hand, since subdividing edges
    does not affect the genus, $\g(G) = \g(G') = \bc'(G') \le
    \bc(G')$.
\end{proof}

\section{FPT Algorithms for Computing $\bco'(G)$ and $\bco(G)$}
\label{sec:main}

We now consider circular layouts, where vertices are placed on a
circle and edges are routed inside the circle. We
note that~$\bco(G)$ and~$\bco'(G)$ can be different.
\begin{observation}
\label{obs:k33}
  $\bco'(K_{3,3})=1$ but $\bco(K_{3,3}) > 1$.
\end{observation}
\begin{proof}
  Let $V(K_{3,3}) = \{a,b,c\} \cup \{a',b',c'\}$. A drawing with
  $\bco'(K_{3,3})=1$ is obtained by placing the vertices
  $a,a',b,b',c,c'$ in clockwise order around a circle; see
  Fig.~\ref{fig:k33}(b).  If a graph~$G$ has $\bco(G) = 1 $ then $G$ is
  planar because we can embed edges for one bundle outside the circle.
  Hence, $\bco(K_{3,3})>1$.
\end{proof}
\begin{figure}[tb]
  \begin{minipage}[b]{0.38\textwidth}
    \centering
    \includegraphics[page=2]{k33}\hfil\includegraphics[page=1]{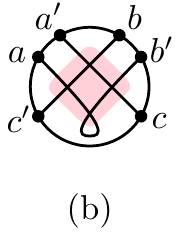}
    \caption{$\bco(K_{3,3}) \ne \bco'(K_{3,3})$; see
      Observation~\ref{obs:k33}}
    \label{fig:k33}
  \end{minipage}
  \hfill
  \begin{minipage}[b]{0.5\textwidth}
    \centering
    \includegraphics{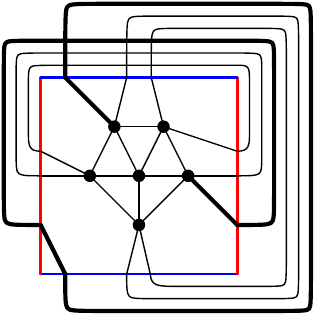}
    \caption{$K_6$ drawn in a fundamental square; the
      self-intersecting edge is bold~\cite[Fig.~2]{afp-bcn-GD16}.}
    \label{fig:k6-tourus}
  \end{minipage}
\end{figure}

Similarly to computing $\bc'(G)$, we use compute $\bco'(G)$ via computing genus.

\begin{theorem}
  \label{thm:bcop-if-fpt}
  Testing whether $\bco'(G)=k$ can be done in $2^{k^{O(1)}}n$ time.
\end{theorem}
\begin{proof}[Sketch]
This follows from the fact that
$\bco'(G) = \g(G^\star)$ where $G^\star$ is a graph with a
vertex~$v^\star$ adjacent to every vertex of~$G$ (see
\lncsarxiv{the full
  version~\cite{full-version-arxiv}}{Lemma~\ref{lem:bcop->genus} in
  Appendix~\ref{sec:missing-lemmas}}
and the $2^{g^{O(1)}}n$ time algorithm for genus~\cite{KawarabayashiMR08}.
\end{proof}

To prove our main result (Theorem~\ref{thm:bco-FPT}) we develop an
algorithm that tests whether $\bco(G)=k$ in FPT time with respect
to~$k$.
\review{R1: What algorithm do you mean here? I think you should start
  the page with smth like "Here we prove Theorem 3..." * looks fine
  for me. JP: looks fine for me.}
Our algorithm is inspired by recent works on circular layouts with at
most $k$ crossings~\cite{Bannister_2014} and circular layouts where
each edge is crossed at most $k$ times~\cite{ckllw-bo-GD17}.
In both of these prior works, it is first observed that the graphs
admitting such circular layouts have treewidth~$O(k)$, and then
algorithms are developed using Courcelle's theorem, which establishes
that expressions in MSO$_2$ logic can be evaluated efficiently.
(For basic definitions of both treewidth and MSO$_2$ logic,
see \lncsarxiv{the appendix of the full
  version}{Appendix~\ref{sec:mso2+courcelle}}.)

\begin{theorem}[Courcelle \cite{courcelle1990,courcelle2012book}]
  \label{thm:courcelle}
  For any integer $t \geq 0$ and any MSO$_2$ formula~$\psi$ of
  length~$\ell$, an algorithm can be constructed which takes
  a graph $G$ with $n$ vertices, $m$ edges, and treewidth at most $t$
  and decides in
  $O(f(t,\ell)\cdot (n+m))$ time whether $G \models \psi$ where the
  function $f$ from this time bound is a computable function of $t$
  and $\ell$.
\end{theorem}

We proceed along the lines of Bannister and
Eppstein~\cite{Bannister_2014}, who used a similar approach to show
that edge crossing minimization in a circular layout is in FPT (as
mentioned in the introduction).
We start by very carefully describing a surface (in the spirit of
Observation~\ref{obs:lift})
onto which we will lift our drawing.
We will then examine the structure of this surface (and our algorithm)
for the case of one bundled crossing and finally for $k$ bundled
crossings.

\subsection{Constructing the Surface Determined by a Bundled Drawing}
\label{sec:constructing-surface}

Consider a bundled circular drawing $D$.
Note that adding parallel edges to the drawing (i.e., making our
graph a multi-graph) allows us to assume that every bundled crossing has four
distinct frame edges and can be done without modifying the number of bundled
crossings; see Fig.~\ref{fig:dwarf-village}.
Each bundled crossing $B$ defines a Jordan curve made up of the
four Jordan
arcs $\tilde e_1$, $\tilde e_2$, $\tilde e_3$, $\tilde e_4$ in clockwise order taken
from its
four frame edges $e_1, \ldots, e_4$  respectively (here $(e_1,e_3)$
and $(e_2, e_4)$ frame the two bundles and $e_i = u_iv_i$).
Similarly to Observation~\ref{obs:lift}, we can construct a surface
$\surfgk$ by creating a flat handle (note that this differs
from the usual definition of a handle since our flat handles have a
boundary)%
\review{R3: l.229: A flat handle is not a handle. A surface with
  handles has no boundary. You create a surface with boundary (where
  the boundary consists of duplicated of the edge pieces $(\tilde
  e)_i, i=1,2,3,4$). Since the bundled crossing number is closely
  related to the oriented genus, it would be good to clarify that flat
  handles are different from handles.}%
\review{R3: lines 229--238: I find the definition of the surface
  $\cal S$ and the subdivision $\Omega$ a bit confusing. I'd prefer a
  more rigourous topological definition of these objects.}
on top of $D$ which connects $\tilde e_2$ to $\tilde e_4$
and doing so for each bundled crossing.
We then lift the drawing $D$ onto $\surfgk$
by rerouting the edges of one of the bundles over its corresponding handle for each bundled crossing
$B$
obtaining the lifted drawing $D_{\surfgk}$.
To avoid the crossings in $D_{\surfgk}$ of the frame edges that can occur at the foot of the handle of
$B$ we can make the handle a bit wider and add \emph{corner-cuts} (as illustrated in Fig.~\ref{fig:one-bc})
to preserve the topology of the surface.
Thus, $D_{\surfgk}$ is crossing-free.

We now cut $\surfgk$ into \emph{components} (maximal
connected subsets) along the frame edges and corner-cuts of each bundled crossing,
resulting in a subdivision~$\surfsub$ of~$\surfgk$.%

We use $D_\surfsub$ to denote the sub-drawing of $D_{\surfgk}$%
\review{R3: At first read, I thought you simply clip the drawing to
  $\Omega$. Emphasize that the drawing of $G$ on $S$ is crossing-free,
  i.e., it is an embedding.}
on $\surfsub$, i.e., $D_\surfsub$ is missing the frame edges
since these have been cut out.
We now consider the components of~$\surfsub$.
Notice that every edge of $D_\surfsub$
is contained in one component of $\surfsub$.
In order for a component $s$ of $\surfsub$
to contain an edge $e$ of $D_\surfsub$, $s$ must have
both endpoints of $e$ on its boundary.
With this in mind we focus on the components of
$\surfsub$ where each one has a vertex of $G$ on its boundary and
call such components \emph{regions}.
Observe that a crossing in $D$ which does not involve a frame edge
corresponds, in $D_\surfsub$, to a pair of edges where one goes
over a handle and the other goes underneath.

\subsection{Recognizing a Graph with One Bundled Crossing}
\label{sec:onebc}

We now discuss how to recognize if an $n$-vertex graph $G={(V, E)}$ can be drawn in a circular layout with one bundled crossing.
Consider a bundled circular drawing~$D$ of~$G$ consisting of
one bundled crossing. The bundled crossing consists of two bundles, and so a set $\fedges$ of
four frame edges.
By $V(\fedges)$ we denote the set of vertices
incident to frame edges.
Via the construction above, we obtain the subdivided surface
$\surfsub$; see Fig.~\ref{fig:one-bc}.
Let $r_1$ and $r_2$ be the regions that are each
bounded by a pair of frame edges corresponding to one of the bundles,
and let $r_3, \dots, r_6$ be %
the regions each bounded by one edge from one pair
and one from the other pair;
see Fig.~\ref{fig:one-bc}.  These are all the regions of~$\surfsub$.
Since, as mentioned before, each of the non-frame edges
of $G$ (i.e., each $e \in E(G) \setminus \fedges$)
along with its two endpoints is contained in exactly one of these regions,
 each component
of $G\setminus V(\fedges)$ including the edges connecting it to vertices of
$V(\fedges)$ is drawn in~$D_\surfsub$ in some
region of~$\surfsub$.
In this sense, for each region $r$ of $\surfsub$, we use $G_r$ to
denote the subgraph of~$G$ induced by the components of $G\setminus
V(\fedges)$ contained in~$r$, %
including the edges connecting them to vertices in $V(\fedges)$.
Additionally, each vertex of $G$ is either incident to an edge in
$\fedges$ (in which case it is on the boundary of at least two
regions) or it is
on the boundary of exactly one region.

\begin{figure}[tb]
  \begin{subfigure}[t]{0.22\textwidth}
    \centering
    \includegraphics[page=2,scale=.8]{region}
    \caption{}
    \label{fig:regions-cprime}
  \end{subfigure}
  \hfill
  \begin{subfigure}[t]{0.22\textwidth}
    \centering
    \includegraphics[page=3,scale=.8]{region}
    \caption{}
  \label{fig:regions}
   \end{subfigure}
  \hfill
  \begin{subfigure}[t]{0.26\textwidth}
    \centering
    \includegraphics[page=1,scale=.8]{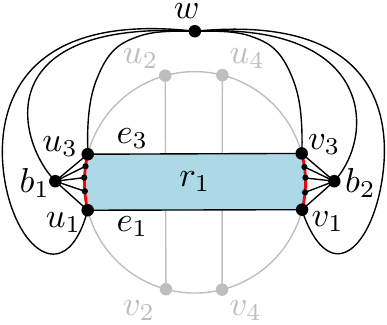}
    \caption{}
    \label{fig:regions-gr}
  \end{subfigure}
  \hfill
  \begin{subfigure}[t]{0.23\textwidth}
    \centering
    \includegraphics[page=4,scale=.8]{region}
    \caption{}
    \label{fig:regions-gr3}
  \end{subfigure}
  \caption{(a) Bundled crossing; (b) regions, corner-cuts in blue;
    (c),(d) augmented graphs $G^*_{r_1}$ and $G^*_{r_3}$ consist of
    the edges of $G_{r_1}$ and $G_{r_3}$ (blue), augmentation vertices
    and edges (black)}
    \label{fig:one-bc}
\end{figure}

Note that there are two types of regions: $\{r_1$, $r_2\}$ and $\{r_3$,
$r_4$, $r_5$, $r_6\}$.  Consider a region of the first type, say~$r_1$; see Fig.~\ref{fig:one-bc}.
Observe that~$r_1$ is a topological disk, i.e., $G_{r_1}$ is outerplanar.
Moreover, $G_{r_1}$ has a special outerplanar drawing where on the boundary
of $r_1$ (in clockwise order) we see  the frame edge~$e_1$, the vertices
mapped to the $(u_1,u_3)$-arc, the frame edge $e_3$, then the vertices
mapped to the $(v_3,v_1)$-arc.
We now describe how to augment $G_{r_1}$ to a planar graph $G^*_{r_1}$
where in every planar embedding of $G^*_{r_1}$ the sub-embedding of $G_{r_1}$
has this special outerplanar form\footnote{This augmentation may sound
overly complicated, but is written as to easily generalize to more bundled crossings.}.
The vertex set of $G^*_{r_1}$ is $V(G_{r_1}) \cup \{h,b_1,b_2\}$ where
we call $h$ \emph{hub}
vertex and $b_1$ and $b_2$ \emph{boundary} vertices (one for each arc
of the boundary of~$r_1$ to which vertices can be mapped); see
Fig.~\ref{fig:one-bc}.  The graph $G^*_{r_1}$ has four types of edges;
the edges in~$E(G_{r_1})$, edges that make~$h$ the hub of a wheel whose
cycle is $C=(v_1,b_2,v_3,u_3,b_1,u_1,v_1)$, edges from~$b_1$ to the
vertices on the $(u_1,u_3)$-arc, and edges from~$b_2$ to the vertices
on the $(v_3,v_1)$-arc (both including the end points).
Clearly, we can obtain a planar embedding of $G^*_{r_1}$ by drawing the
elements of $G^*_{r_1} \setminus G_{r_1}$ ``outside'' of the outerplanar drawing of $G_{r_1}$ described before.
Moreover, every planar embedding of $G^*_{r_1}$ contains an outerplanar
embedding of $G_{r_1}$ that can be drawn in the special form needed to
``fit'' into $r_1$, in the sense that all of $G_{r_1}$ lies (or can be
put) inside the simple cycle~$C$.  (For example, if, say, $b_1$ is a
cut vertex, the component hanging off~$b_1$ can be embedded in the
face $(h,b_1,u_3,h)$.  But then it can easily be moved into~$C$.
Similarly, a component that is incident only to~$u_3$ and~$v_3$ can
end up in the face $(h,u_3,v_3,h)$, but again, the component can be
moved inside~$C$.)

\review{R3: This might be a good place to introduce the term
  "region graph" used later.  AW: Not used anymore.}
Similarly, for a region of the second type, say $r_3$,
the graph $G_{r_3}$ is outerplanar with a special drawing where
all the vertices must be on the $(u_3, u_2)$-arc of the disk subtended by
the two frame edges $e_3$ and $e_2$ bounding the region $r_3$.
We augment similarly as for $r_1$; see Fig.~\ref{fig:one-bc}.
For the augmented graph $G^*_{r_3}$, we add to $G_{r_3}$
a boundary vertex~$b$ neighboring all vertices on the
$(u_3,u_2)$-arc and a hub vertex~$h$ adjacent to~$u_2$, $b$,
and~$u_3$.  Again, $G^*_{r_3}$ is planar since $G_{r_3}$ is
outerplanar due to $r_3$ being a topological disk.  Moreover, as $b$
is adjacent to all vertices of~$G_{r_3}$, in every planar embedding
of~$G^*_{r_3}$, $G_{r_3}$ is embedded outerplanarly and, since $b$
occurs on one side of the triangle $u_3u_2h$, the edge $u_3u_2$
occurs on the boundary of this outerplanar embedding of~$G_{r_3}$.
Thus, each planar embedding of $G^*_{r_3}$ provides an outerplanar
embedding of $G_{r_3}$ that fits into~$r_3$.

Note that each $G_{r_i}$ fits into $r_i$ because its
augmented graph $G^*_{r_i}$ is planar~($\star$).
Moreover, as outerplanar graphs have treewidth
at most two~\cite{m-latroamog-IPL79}, each graph~$G_r$ is outerplanar,
and adding the (up to) eight frame vertices raises the treewidth by at
most~8, we see that the treewidth of~$G$ is at most~10.
Namely, in order for $G$ to have $\bco(G)=1$, it must have treewidth
at most~10 (and this can be checked in linear time using an algorithm
of Bodlaender~\cite{bodelaender96}).

To sum up, $G$ has a circular drawing $D$ with at most one bundled
crossing because it has treewidth at most~10 and there exist (i)~$\beta \le 4$
frame edges $e_1, e_2, \dots, e_\beta$ (this set is denoted $\fedges$) and $v_1, \ldots, v_\xi$ frame vertices (this set is denoted $\fvertices$),
(ii)~a particular circular drawing $D_\fedges$ of frame edges,
(iii)~the drawing of the one bundled crossing $B$,
and (iv)~$\gamma \le 6$ corresponding regions $r_1, \ldots, r_\gamma$
of the subdivided surface $\surfsub$
so that the following properties hold (note that the frame vertices
partition the boundary of the disk underlying $\surfsub$ into $\eta
\le 8$ (possibly degenerate) arcs $p_1, \ldots, p_\eta$ where each
such $p_j$ is contained in a unique region $r_{i_j}$ of $\surfsub$):
\begin{enumerate}
\label{msoprops}
\item %
$E(G)$ is partitioned into $E_0, E_1, \dots, E_\gamma$, where $E_0{=}\{f_1, \ldots, f_\beta\}$.
\label{property0-1}
\item %
$V(G)$ is partitioned into $V_0, V_1, \ldots, V_\eta$, where $V_0{=}\{u_1, \ldots, u_\xi\}$.
\label{property0-2}
\item The mapping $u_i \leftrightarrow v_i$ and $f_i \leftrightarrow e_i$ defines an isomorphism between the subgraph of $G$ formed by $(V_0,E_0)$ and graph  $(\fvertices,\fedges)$.
\label{property1}
  \item No vertex in $V(G) \setminus V_0$
    has incident edges $e \in E_i$, $e' \in E_j$ for~$i \neq j$.
    \label{property2-1}
  \item For each $v \in V_0$, and each
    edge $e$ incident to $v$, exactly one of the following is true:
    (i) $e\in E_0$ or (ii) $e \in E_i$ and $v$ is on the boundary of~$r_i$.
    \label{property2-2}
  \item For each $v \in V_j$, all edges incident to $v$ belong to $E_{i_j}$.
  \label{property2-3}
  \item For each region $r_i$,
  let $G_i$ be the graph $(V_0 \cup \bigcup_{j \colon i_j=i}V_j,E_i)$
  (i.e., the subgraph that is to be drawn in~$r_i$),
  and let $G^*_i$ be the corresponding augmented graph
  (i.e., as in $\star$ above).  Each $G^*_i$ is planar.
  \label{property3}
\end{enumerate}

We now describe the algorithm to test for a simple circular drawing with one bundled crossing.
First we check that treewidth of $G$ is at most~10.
We then enumerate drawings of up to four edges in the circle.
For the drawing~$D_\fedges$ that is valid for the set~$\fedges$
of frame edges of one bundled crossing,
\review{R3: This might make sense for k bundled crossings, but
  for one bundled crossing I can think of only one valid simple
  drawing: where all crossings between $e_1,...,e_4$ are on the
  boundary of the regions associated with the bundled crossing.}
we define our surface and its regions (which
makes the augmentation well-defined).
We have intentionally phrased these properties
so that it is clear that they are expressible in MSO$_2$
(see \lncsarxiv{the full
  version~\cite{full-version-arxiv}}{Appendix~\ref{sec:formula}}).
The only property that is not obviously expressible is the planarity
of~$G^*_i$.  To this end, recall that planarity is characterized by
two forbidden minors (i.e., $K_5$ and $K_{3,3}$)
and that, for every fixed graph $H$, there is an MSO formula
$\textsc{minor}_H$ so that for all graphs $G$, it holds that
$G \models \textsc{minor}_H$ if and only if $G$ contains $H$ as a
minor~\cite[Corollary 1.14]{courcelle2012book}.
Additionally, each $G^*_i$ can be expressed as an
\emph{MSO-transduction}\footnote{For the formalities of
  transductions, see the book of Courcelle and
  Engelfriet~\cite[Section 1.7.1, and Definitions 7.6 and
  7.25]{courcelle2012book}.} of $G$ and our variables (our
transduction can be thought of as a kind of 2-copying transduction).
Thus, by~\cite[Theorem 7.10]{courcelle2012book} using the transduction
and the MSO formula testing planarity, we can construct an MSO$_2$
formula $\iota$ so that when $G \models \iota$, $G^*_i$ is planar for
every $i$.
Therefore, Properties~\ref{property0-1}--\ref{property3} can be
expressed as an MSO$_2$ formula~$\psi$ and,
by Courcelle's theorem, there is a computable function~$f$
such that we can test (in $O(f(\psi,t)n)$ time) whether $G\models\psi$
for an input graph~$G$ of treewidth at most~$t$.
Thus, since our graph has treewidth at most~10, applying Courcelle's theorem completes our algorihtm.

\subsection{Recognizing a Graph with $k$ Bundled Crossings}
\label{sec:kbc}

We now generalize the above approach to $k$
bundled crossings.  In a drawing $D$ of $G$ together with a solution
consisting of~$k$ bundled crossings, there are $2k$ bundles making (up
to) $4k$ frame edges~$\fedges$.
As described above, these bundled crossings provide a surface $\surfgk$, its subdivision
$\surfsub$,
and the corresponding set of regions. The key ingredient above was that every region was
a topological disk. However, that is now non-trivial as our regions can go over
and under many handles.
To show this property, we first consider the following two partial
drawings $D_A(p)$ and $D_B(p)$ of a matching with
$p+1$ edges $f_0, f_1 \dots, f_p$  (see, e.g.,
Fig.~\ref{fig:conf-ex}) such that
\begin{itemize}[topsep=0pt]
\item edge $f_i$ crosses only $f_{i-1 \bmod p+1}$ and $f_{i+1 \bmod p+1}$
  for $i=0, \dots, p$;
\item the endpoints of each edge $f_i$, $i=1, \dots, p-2$, are
  inside the cycle $C$ formed by the crossing points and the edge-pieces
  between these crossing points;
\item both endpoints of $f_{p-1}$, only one endpoint of $f_0$, and only
  one endpoint of $f_p$ are contained in~$C$ in the drawing
  $D_A(p)$;
\item only one endpoint of $f_{p-1}$, only one endpoint of $f_0$, and
  no endpoints of $f_p$ are contained in~$C$ in the drawing
  $D_B(p)$.
\end{itemize}
Note that the partial drawings $D_A(p)$ and $D_B(p)$ differ only in how the last
edge is drawn with respect to the previous edge.
Arroyo et al.~\cite[Theorem 1.2]{edgap-abr-arXiv18} showed that such
partial drawings
are obstructions for pseudolinearity, that is,
they cannot be part of any pseudoline arrangement.
Therefore, neither of these partial drawings can be \emph{completed}
to a simple circular drawing, that is,
the endpoints of the edges cannot be extended so that they lie on a
circle which contains the drawing. We highlight this fact in the
following lemma.
\review{R3: Lemma 1 seems to follow from a result by Arroyo, Bensmail,
  and Richter: "Extending Drawings of Graphs to Arrangements of
  Pseudolines" arXiv:1804.09317}
\setCounters{lemma}
\wormhole{lemconf}
\newcommand{\lemconf}{%
  For a matching with $p+1$ edges $f_0, f_1, \dots, f_p$, neither the
  partial drawing $D_A(p)$ nor $D_B(p)$ can be completed to a simple
  circular drawing.}
\begin{lemma}
\label{lem:conf}
\lemconf
\end{lemma}
\begin{figure}[tb]
  \begin{subfigure}[b]{0.32\textwidth}
    \centering
    (a)
    \includegraphics[page=1]{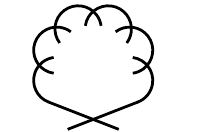}
  \end{subfigure}
  \hfil
  \begin{subfigure}[b]{0.32\textwidth}
    \centering
    (b)~~
    \includegraphics[page=2]{conf-ex}
  \end{subfigure}
  \hfil
  \begin{subfigure}[b]{0.32\textwidth}
    \centering
    (c)
    \includegraphics[page=3]{conf-ex}
  \end{subfigure}
  \caption{Configurations for $p=6$: (a)~$D_A(p)$,
    (b)~$D_B(p)$, and (c)~induced by~a~hole}
  \label{fig:conf-ex}
\end{figure}

Using this lemma we can now prove the following statement.

\begin{lemma}
  \label{lem:region-planar}
    Each region $r$ of $\surfsub$ is a topological
    disk\footnote{We slightly abuse this
    notion to also mean a simply connected set.}.
\end{lemma}

\begin{proof}
  First, we show that no region of $\surfsub$
  includes part of both a handle and its
  \emph{undertunnel}, that is, the part of the surface over which the handle was built.
  Then we will show that a region also does not include holes.

  Let $r$ be a region of the surface subdivision $\surfsub$.
  The boundary of this region is formed by pieces of frame edges that were lifted on
  the surface $\surfgk$ as described above and the additional corner-cuts
  as illustrated in Figure~\ref{fig:one-bc} in red.
  Consider the projection $r'$ of $r$ and its boundary
  on the drawing $D$ in the plane.
  Note that the projected boundary
  either follows an edge in $D$ or switches to some another edge
  via a corner-cut at an intersection point; see Fig.~\ref{fig:main-lemma}(a).

  Suppose now, for a contradiction, that $r$ contains both a handle and
  its undertunnel corresponding to the same bundled crossing $B=((e_1,
  e_3), (e_2, e_4))$.
  \review{R3: The word "include" suggest containment. But as far as I
    understand, the boundaries of regions may contain edges that go
    through the upper or lower part of a flat handle. Do you mean that
    $r$ contains part of both the undertunnel and the overpass. Also,
    even if $r$ contains a Jordan arc that traverses both parts of a
    flat handle, $r$ could still be simply connected.  You prove that
    $r$ does not contain any Jordan arc whose projection is has a
    self-intersection---please clarify how this is related to
    noncontractible cycles in $r$.}
  Then there is a Jordan arc~$\gamma \subset r$ going over and under this
  handle making a loop; see Fig.~\ref{fig:main-lemma}(b).
  Note that the orthogonal projection $\gamma'$ of $\gamma$ on the disk of the drawing $D$
  self-intersects. The profile of edges along the
  projected boundary of $r$ that is enclosed by $\gamma'$
  then inevitably contains a partial drawing $D_A(p)$; see Fig.~\ref{fig:main-lemma}(c).
  And according to Lemma~\ref{lem:conf},
  such a partial drawing cannot be completed
  to a valid simple circular drawing; contradiction.

   As for holes, it is easy to see that if $r$ had a hole, the profile of the boundary
  edges around this hole would %
  give a partial drawing of edges as illustrated in Fig.~\ref{fig:conf-ex}(c).
  Therefore, the region $r$ is a proper
  topological disk.
\end{proof}

\begin{figure}[tb]
  \begin{subfigure}[b]{0.2\textwidth}
    \centering
    \includegraphics{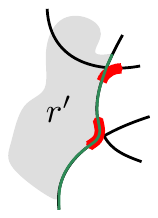}
    \caption{}
  \end{subfigure}
  \hfill
  \begin{subfigure}[b]{0.2\textwidth}
    \centering
    \includegraphics{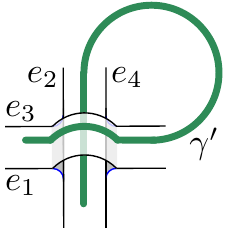}
    \caption{}
  \end{subfigure}
  \hfill
  \begin{subfigure}[b]{0.4\textwidth}
    \centering
    \includegraphics{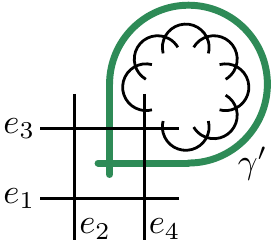}
    \caption{}
  \end{subfigure}
  \caption{%
      (a)~Projection $r'$ of the region $r$ and its boundary (green, the corner-cuts are in blue)
      onto the disk of the drawing $D$
      (b)~projection~$\gamma'$ of a Jordan arc~$\gamma$ that goes over and under the same
      handle;
      (c)~profile of edges of the projected boundary of $r$ enclosed by the loop made by $\gamma'$
      form a partial drawing $D_A(p)$.
      }
   \label{fig:main-lemma}
\end{figure}

The next lemma concerning treewidth is
a direct consequence of Lemma~\ref{lem:region-planar}.

\newcommand{\lemtreewidth}{%
  \label{lem:treewidth}
  If a graph $G$ admits a circular layout with $k$ bundled crossings
  then its treewidth is at most $8k+2$.}
\begin{lemma}
  \lemtreewidth
\end{lemma}
\begin{proof} %
  If the graph $G$ can be drawn in a circular layout with $k$ bundled
  crossings then there exist at most $4k$ frame edges.  According to
  Lemma~\ref{lem:region-planar}, the removal of their endpoints breaks
  up the graph into outerplanar components.  The treewidth of an
  outerplanar graph is at most two~\cite{m-latroamog-IPL79}.
  Moreover, adding a vertex to a graph raises its treewidth by at
  most one.  Thus, since deleting at most $8k$ frame vertices
  leaves behind an outerplanar graph, $G$ has treewidth at most
  $8k+2$.
\end{proof}

We now prove Theorem~\ref{thm:bco-FPT}, that deciding
whether $\bco(G) \le k$ is FPT in $k$.

\begin{proof}[of Theorem~\ref{thm:bco-FPT}]
  We use Lemma~\ref{lem:region-planar} and extend the algorithm of
  Section~\ref{sec:onebc}.

  Suppose $G$ has a circular drawing $D$ with at most $k$ bundled
  crossings.
  In $D$ we see the set $\fedges$ of (up to) $4k$
  frame edges of these bundled crossings.
  As before,
  $\fedges$ together with $D$ defines a subdivided topological
  surface $\surfsub$ containing a set of regions
  $R$.  As in the one bundled crossing case,
  each edge of $G$ is %
  in exactly one such region,
  and each
  vertex of $G$ either is incident to an edge in $\fedges$ (in
  which case it belongs to at least two regions) or belongs to
  exactly one region.

  \begin{figure}[tb]
    \begin{subfigure}[t]{0.49\textwidth}
      \centering
      \includegraphics[page=1,scale=.95]{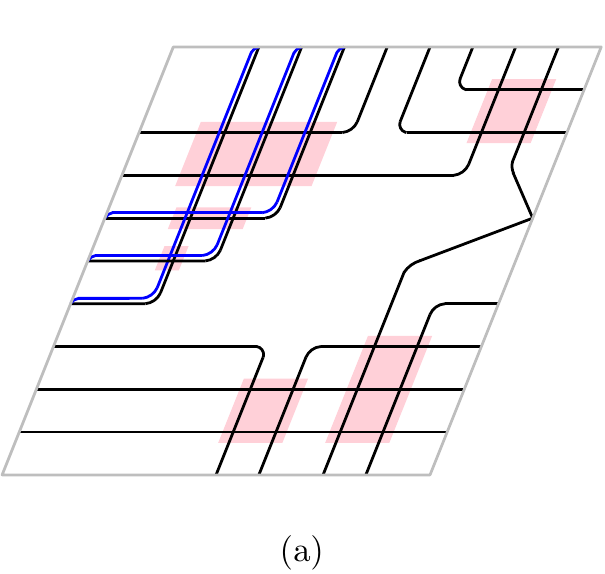}
    \end{subfigure}
    \begin{subfigure}[t]{0.49\textwidth}
      \centering
      \includegraphics[page=2,scale=.95]{dwarf-village-modern-sheared-final}
    \end{subfigure}
    \caption{%
      (a) A bundled drawing $D$ with six bundled crossings (pink);
      parallel (blue) edges can be inserted to avoid degenerate
      bundled crossings; (b) the corresponding surface of genus 6; the
      components of the surface that are not regions are marked in
      green; the region $r$ (light blue) has a boundary consisting of
      the arcs of the disk (red) and the arcs $c_1$, $c_2$, $c_3$, and
      $c_4$ (traced in orange).}
     \label{fig:dwarf-village}
  \end{figure}

  Throughout the proof we will refer to Fig.~\ref{fig:dwarf-village} for an
  example.
  By Lemma~\ref{lem:region-planar}, each region $r$ is a topological
  disk and as such its graph
  $G_r$ is outerplanar with a quite special drawing $D_r$ described as follows.
  In particular, if we
  trace the boundary of $r$ in clockwise order, we see that it is made
  up of arcs $p_1, \ldots, p_\alpha$ of~$\surfgk$,
  marked in red in Fig.~\ref{fig:dwarf-village}(b)
  (such arcs can degenerate to single points),
  and Jordan arcs $c_1, \ldots, c_\alpha$,
  traced in orange in Fig.~\ref{fig:dwarf-village}(b),
  each of which connects two such arcs of the disk.
  For $i \in \{1,\dots,\alpha\}$,
  let~$u_i$ and~$u'_i$ be the end points of~$p_i$, in clockwise order.
  So $u'_i$ and $u_{i+1}$ are the endpoints of~$c_i$.
  No vertex of $G_r$ lies in the interior of~$c_i$.

  We now describe $G^*_r$.  First, we add a hub vertex~$h$.  Then, for
  each $i \in \{1,\dots,\alpha\}$, if~$u'_i$ and~$u_{i+1}$ (where
  $u_{\alpha+1}$ is~$u_1$) are not adjacent, we add an edge between
  them.  If~$p_i$ is non-degenerate, we add  
  a boundary vertex~$b_i$ adjacent to all vertices on~$p_i$
  (including~$u_i$ and~$u'_i$) and make $h$ adjacent to~$u_i$, $b_i$,
  and~$u_i'$.  Otherwise, we make $h$ adjacent to $u_i=u_i'$ and,
  for technical reasons (see \lncsarxiv{the full
    version}{Appendix~\ref{sec:formula}}), we identify~$b_i$
  with~$u_i$ and~$u_i'$.

  Observe that the resulting graph~$G^*_r$ is planar due to the
  special outerplanar drawing of~$G_r$ in~$r$.
  Moreover, in every planar embedding of $G^*_r$, there is an
  outerplanar embedding of $G_r$ where the cyclic order of the arcs
  $c_i$ and the sets of vertices mapped to the $p_i$'s match their
  cyclic order in $r$, implying that $G_r$ fits into $r$.
  This is due to the fact that the simple cycle~$C'$ around~$h$
  must be embedded planarly, with all of~$G_r$ inside (with the
  possible and easy-to-fix exceptions described in
  Section~\ref{sec:onebc} concerning the cycle~$C$ there).  Then the
  order of the vertices in an outerplanar embedding of~$G_r$
  is the order of the vertices incident to $b_1,\dots,b_\alpha$
  in a planar embedding of~$G_r^*$.
  \review{R4, l.11: why does this force the cyclic order of frame
    edges to match how they occur in $r$?}
  So the planarity of~$G^*_r$ guarantees that~$G_r$ fits into~$r$ as needed.

  The reason why $G$ has a circular drawing $D$ with at most $k$
  bundled crossings is that there is a $\beta$-edge $k$-bundled
  crossing drawing~$D_\fedges$ (of the graph formed by~$\fedges$),
  whose corresponding surface $\surfgk$ consists of regions $r_1,
  \dots, r_\gamma$ (note: $\gamma \le 2\beta \le 8k$) so that
  Properties~\ref{property0-1}--\ref{property3} hold.

  Our algorithm first checks that the treewidth of $G$ is at most
  $8k+2$.  Recall that this can be done in linear time (FPT
  in~$k$)~\cite{bodelaender96}.  It then enumerates all
  possible simple drawings of at most $4k$ edges in the circle\footnote{i.e., at most $4k$ curves extending to infinity in both directions where each pair of curves cross at most once. The number of such drawings is proportional to $k$, and efficient enumeration has been done for the case when every pair of curves cross exactly once~\cite{felsner}.}.
  \review{R3: Pseudolines pairwise cross each other, but a simple 1-page drawing
  can have disjoint edges. So you need to include additional arrangements.}
  For each drawing, it further enumerates
  the possible ways to form $k$ bundled crossings so that every edge is a frame edge of at least one
  bundled crossing.
  Then, for each such bundled drawing $D_\fedges$, we build an
  MSO$_2$ formula~$\varphi$ (see
  \lncsarxiv{the full version}{Appendix~\ref{sec:formula}}) to
  express Properties~\ref{property0-1}--\ref{property3}.
  Finally, since $G$ has treewidth at most $8k+2$, we can apply
  Courcelle's theorem on $(G,\varphi)$.
\end{proof}

\section{Open Problems}
\label{sec:open}

Given our new FPT algorithm for simple circular layouts, it would be
interesting to improve its runtime and to investigate whether a
similar result can be obtained for general simple layouts.
A starting point could be the FPT algorithm of Kawarabayashi et
al.~\cite{Kawarabayashi_2007} for computing
the usual crossing number of a graph.

\paragraph{Acknowledgements.}
We thank Bruno Courcelle for clarifying discussions on the tools
available when working with his meta-theorem and in particular MSO$_2$.

\bibliographystyle{splncs04}
\bibliography{abbrv,bundled_crossings}

\lncsarxiv{\end{document}}{}

\clearpage
\appendix

\section*{Appendix}

\section{Definitions: Treewidth and MSO$_2$}
  \label{sec:mso2+courcelle}

  The purpose of this subsection is to provide the necessary
  definitions (i.e., treewidth and MSO$_2$) needed for Courcelle's
  theorem; see Theorem~\ref{thm:courcelle}.

  The concept of \emph{treewidth} was introduced by Robertson and
  Seymour~\cite{robertson1984graph}.  A {\em tree decomposition} of a
  graph $G$ is a pair $({X}, T)$, where $T$ is a tree and ${X}=\{{X}_i
  \mid i\in V(T) \}$ is a family of subsets of $V(G)$, called
  \emph{bags}, such that (1) for all $v \in V(G)$, the set of nodes
  $T_v = \{i \in V(T) \mid v \in {X}_i\}$ induces a non-empty
  connected subtree of $T$, and (2) for each edge $uv \in E(G)$ there
  exists $i \in V(T)$ such that both $u$ and $v$ are in ${X}_i$.  The
  maximum of $|{X}_i|-1$, $i\in V(T)$, is called the {\em width} of
  the tree decomposition.  The {\em treewidth}, $tw(G)$, of a graph
  $G$ is the minimum width over all tree decompositions of $G$.  For
  our purposes, an important fact is that every outerplanar graph $G$
  has $tw(G) \le 2$~\cite{m-latroamog-IPL79}.

  \emph{Extended Monadic Second-Order Logic} (MSO$_2$) is a subset of
  \emph{second-order logic} that can be used to express certain graph
  properties.  It is built from the following primitives:
  \begin{itemize}
  \item variables for vertices, edges, sets of vertices, and sets of
    edges;
  \item binary relations for: equality ($=$), membership in a set
    ($\in$), subset of a set ($\subseteq$), and edge--vertex incidence
    ($I$);
  \item standard propositional logic operators: $\lnot$, $\land$,
    $\lor$, $\rightarrow$, and $\leftrightarrow$;
  \item standard quantifiers ($\forall,\exists$) which can be applied
    to all types of variables.
  \end{itemize}
  Note that, if we drop the ``$_2$'' then we have \emph{Monadic Second-Order Logic} (MSO)
  where the only difference is that we are now not allowed to quantify over edge sets.

  For a graph $G$ and an MSO$_2$ formula~$\psi$, we use $G \models
  \psi$ to indicate that $\psi$ can be satisfied by $G$ in the obvious
  way.

\section{Missing Proofs of Section~\ref{sec:main}}
\label{sec:missing-lemmas}

\begin{lemma}
  \label{lem:bcop->genus}
  Given a graph $G = (V,E)$, let $G^\star$ be the graph obtained
  from~$G$ by adding a new vertex~$v^\star$ adjacent to every vertex
  of~$G$.  Then $\bco'(G) = \g(G^\star)$.
\end{lemma}

\begin{proof}
  Similarly as in \cite[Theorem~1]{afp-bcn-GD16}, it is easy to see
  that $\bco'(G)$ is an upper bound for the genus of $G^\star$, because,
  according to Observation~\ref{obs:lift}, we can lift any circular
  drawing of $G$ onto a surface $\surfgk$ of genus $\bco'(G)$
  and then we can add $v^\star$ using the outside of the circle.
  Clearly, this produces a crossing-free drawing of $G^\star$ on the
  surface $\surfgk$.

  It remains to show that given a crossing-free drawing of $G^\star$ on a
  surface of genus $k$, we can construct a circular drawing of $G$
  with at most $k$ bundled crossings.  Consider a drawing of $G^\star$ on
  a surface $\surfgk$ of genus $k$;
  see Fig.~\ref{fig:gprime} for instance.
  \begin{figure}[tb]
    \begin{subfigure}[b]{0.32\textwidth}
      \centering
      \includegraphics[page=1]{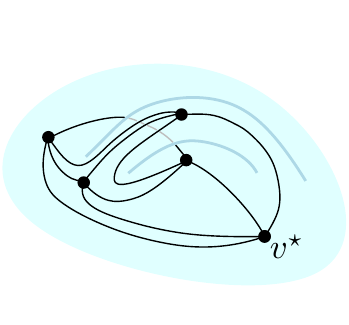}
      \caption{the graph $G^\star$}
      \label{fig:gprime}
    \end{subfigure}
    \hfill
    \begin{subfigure}[b]{0.32\textwidth}
      \centering
      \includegraphics[page=2]{bcop-fpt}
      \caption{modifying the drawing}
      \label{fig:drag}
    \end{subfigure}
    \hfill
    \begin{subfigure}[b]{0.32\textwidth}
      \centering
      \includegraphics[page=3]{bcop-fpt}
      \caption{fundamental polygon}
      \label{fig:polygon}
    \end{subfigure}
    \caption{
      Obtaining a circular drawing with $k$ bundled crossings of
      $G$ from the embedding of $G^\star$ on a surface of genus $k$.}
  \end{figure}
  We can modify the drawing
  so that all the neighbors $N(v^\star)$ of $v^\star$ in $G^\star$
  are placed in an $\epsilon$-neighborhood of~$v^\star$
  in~$\surfgk$ %
  (which is a
  topological disk).  We now explain the modification in more detail.
  Consider all the edges incident to $v^\star$ in the drawing and
  drag each neighbor~$u$ of~$v^\star$ along the edge~$uv^\star$ (as
  illustrated in Fig.~\ref{fig:drag}) until it reaches the
  $\epsilon$-neighborhood on the surface $\surfgk$.  Since for
  each $u\in N(v^\star)$ the edges $uw \in E$ with $w \neq v^\star$
  are bundled together at the position where $u$ was on the surface and
  dragged together with~$u$ along the edge~$uv^\star$, this does not
  introduce any crossings.  Then we use the fundamental polygon
  representation~\cite[Theorem~1]{afp-bcn-GD16} to the modified
  drawing of $G^\star$ on the surface $\surfgk$ of genus
  $k$. Since all the vertices are located on the boundary of the
  $\epsilon$-neighborhood of~$v^\star$ (which itself is a
  surface of genus~0), there exist a representation where all
  edges between~$v^\star$ and $V\setminus v^\star$ are drawn
  inside the polygon.  After removing the vertex~$v^\star$ from the
  representation, we obtain a circular drawing of~$G$ with
  at most $k$ bundled crossings.
\end{proof}

\section{MSO$_2$ Formula for Testing Whether $\bco(G) \le k$}
\label{sec:formula}

For the class of formulas expressible in MSO$_2$, we refer to
Appendix~\ref{sec:mso2+courcelle}; see also the textbook of Courcelle
and Engelfriet~\cite{courcelle2012book} for more background.
We now construct an MSO$_2$ formula to express the following problem:
\begin{itemize}
\item Given a graph $G=(V, E)$ and a simple circular drawing $D_{\fedges}$
with $k$ bundled crossings so that $\fedges =\{e_1, \ldots, e_\beta\}$ is the set of frame edges (and $D_{\fedges}$ has no other edges) and $\fvertices =\{v_1, \ldots, v_\xi\}$ is the set of frame vertices (and $D_{\fedges}$ has no other vertices);
\item determine whether $G$ has a simple circular drawing with $k$ bundled crossings so that the frame edges and vertices occur as in $D_{\fedges}$.
\end{itemize}
This is based on Properties~\ref{property0-1}--\ref{property3} on page~\pageref{msoprops}:
we express them as
MSO$_2$ formulas.

Properties~\ref{property0-1}~and~\ref{property0-2} simply state that
a set of elements is partitioned into a certain number of disjoint subsets.
We use a formula stated by Bannister and Eppstein~\cite{Bannister_2014} to
express this in MSO$_2$.
For example, partitioning of a set~$E$ into~$E_0, E_1, \dots, E_\gamma$  disjoint subsets
can be done in the following way.

\[
\textsc{Partition}(E; E_0, \ldots, E_\gamma) =
(\forall e\in E) \big[
\big(
\bigvee_{i=0}^{\gamma} e \in E_i
\big)
\wedge
\big(
\bigwedge_{i \neq j} \neg (e \in E_i \wedge e \in E_j)
\big)
\big].
\]

We will additionally use the following formula to state that a vertex set $V'$
is the set of endpoints of an edge set $E'$:
\[
\textsc{Incident}(V', E') ~=~
(\forall e \in E')~(\forall v \in V(G))~
[ I(e, v) \Leftrightarrow v \in V' ].
\]

We now turn to the properties more specific to our fixed drawing $D_{\fedges}$ of $\beta \le 4k$ frame edges $\fedges=
\{f_1, f_2, \dots,  f_{\beta} \}$ whose endpoints are
$V(\fedges) = \{u_1, u_2, \dots,  u_{\xi} \}$, where $\xi \le 2\beta$.
As discussed in Section~\ref{sec:constructing-surface} and
Lemma~\ref{lem:region-planar}, this drawing induces a corresponding set of
regions $r_1, \dots, r_{\gamma}$.

Property~\ref{property1} ensures that certain
edges~$E_0 = \{e_1, e_2, \dots, e_{\beta}\}$ and
their endpoints~$V_0 = \{v_1, v_2, \dots, v_{\xi}\}$ of the graph~$G$
induce a graph isomorphic to~$(V(\fedges), \fedges)$. This can be modeled by the following
formula.
\begin{align*}
\theta_3(\{v_1, v_2, \dots, v_{\xi}\}, E_0) ~=~   & \big(\forall i, j\in
\{1, 2, \dots, \xi\}\big) ~ \\
\Big[
& \big( (\exists e \in E_0) ~ I(e, v_i) \wedge I(e, v_j) \big) \Leftrightarrow \\
& \big( (\exists f \in \fedges) ~ I(f, u_i) \wedge I(f, u_j) \big)
\Big].
\end{align*}

To express Properties~\ref{property2-1}~and~\ref{property2-2} we
introduce some helpful notation.
We denote the set of boundary vertices of the region
$r_i$ as $\partial r_i$, where $\partial r_i$ is ordered cyclically as in $D_\fedges$.
For example, for the one bundled crossing
case in Fig.~\ref{fig:one-bc}, $\partial r_1 = \{u_1, u_3, v_3,
v_1\}$ and $\partial r_3 = \{u_3, u_2\}$.
For each vertex $v_i\in V(\fedges)$, $i =1, 2, \dots, \xi$, we denote the
indices of regions incident
to $v_i$ in the drawing $D_{\fedges}$ as
$\sigma(i)$,
that is, $\sigma(i)=\{j \mid  v_i \in \partial r_j\}$.
Then Properties~\ref{property2-1} and \ref{property2-2} can be
expressed in MSO$_2$ as follows:

\[
\theta_4(V_0) ~=~ \neg \big(\exists v\in V(G)\setminus V_0\big)~
\Big[\bigvee_{i\neq
j} e_i\in E_i \wedge e_j\in E_j \wedge I(e_i, v) \wedge I(e_i,
v)\Big],
\]

\begin{align*}
  \theta_5(\{v_1, v_2, \dots,  v_{\xi}\},E_0)
  ~=~& \big(\forall i \in \{1, 2, \dots, \xi\}\big)~(\forall e\in E) \\
  & \Big[I(e, v_i) \Rightarrow \big[e\in E_0 \vee (\exists j \in
    \sigma(i)) \left[e \in E_j\right]\big]\Big].
\end{align*}

Finally, we turn to Property~\ref{property3}.
First, note that testing planarity of a graph~$G$ can be expressed as
follows where the formula for \textsc{Minor}$_H(G)$ does not need edge
set quantification (i.e., it is in MSO) \cite[Corollaries 1.14 and
1.15]{courcelle2012book}:
\[
\textsc{Planar}(G) ~=~ \neg \textsc{Minor}_{K_5}(G)
\wedge \neg\textsc{Minor}_{K_{3,3}}(G).
\]

Now, we describe the MSO-transduction\footnote{Note that, a \emph{transduction} is essentially just the name for the operation of constructing the model of one graph/structure from the model of another graph/structure in the language of MSO.}
$\tau_i$ of $G$ to $G^*_i$ (for each region $r_i$; see
Section~\ref{sec:kbc}) subject to the variables $v_1, \ldots, v_\xi,
V_1, \ldots, V_\eta, e_1, \ldots, e_\beta, E_1, \ldots, E_\gamma$.
Note that in our transduction, the input uses the format allowing for edge set quantification (i.e., where we have the objects $V \cup E$ and the binary incidence function $I$), but our output involves the format without edge set quantifications (i.e., where we have the objects $V$ and the binary adjacency function $adj$).
Recall that $\sigma(i) = \{j_1, \ldots, j_{\zeta}\}$ denotes the indices of the frame vertices incident to $r_i$ and is ordered cyclically as in $D_{\fedges}$.
Further, let $V_{l_1}, \ldots, V_{l_{\alpha}}$ be the sets corresponding to the arcs of the boundary of $r_i$ (in order).
With this notation, we can now set up the transduction $\tau_i$ which describes our graph $G^*_i$ in terms of our variables (note that in the statement of \cite[Theorem 7.10]{courcelle2012book} our variables are the \emph{parameters}).
Note that the symbols $h, b_1, b_2, \ldots, b_\alpha$ are new objects that are added in the construction (namely, the hub and boundary vertices of $G^*_i$).
Further, let $C$ be the cycle of the wheel.  Then $V(C)=\{v_{j_1},
\dots, v_{j_{\zeta}}, b_1, \dots, b_{\alpha} \}$.
For each vertex $x \in V(C)$, let $N_C(x)$ be the set
consisting of the two neighbors of~$x$ in~$C$.

\smallskip
\noindent
\underline{The transduction $\tau_i$:}

\noindent
$\displaystyle V(G^*_i) := \; \{h\} \cup V(C) \cup
\bigcup_{j=1}^{\alpha} V_{l_j} $;

\vspace{-\baselineskip}

\begin{align*}
  adj_{G^*_i}(u,v) := (u \ne v) ~\wedge ~\Big(
  & \big( (\exists e \in E_i) ~ (I(e,v) \wedge I(e,u))\big)\\
  & \vee \big((h = u) \wedge (v \in V(C)) \big)
    \vee \big((h = v) \wedge (u \in V(C)) \big) \\
  & \vee \left( \bigvee_{j=1}^{\alpha} u = b_j \wedge v \in
    V_{l_j} \right) \vee \left(\bigvee_{j=1}^{\alpha} v = b_j
    \wedge u \in V_{l_j} \right) \\ 
  & \vee \big( (u \in V(C)) \wedge (v \in N_C(u)) \big) \Big). 
\end{align*}
With this transduction $\tau_i$ and the expression $\textsc{Planar}(G)$, we can now apply~\cite[Theorem 7.10]{courcelle2012book} to obtain the MSO$_2$ formula $\iota_i$ which when applied to $G$ and our parameters allows us to express that $G^*_i$ is planar.
Namely, by taking the conjugation of all of these $\iota_i$ we obtain the needed MSO$_2$ formula $\iota$ (which can be applied to $G$ and our variables) to express that all of the $G^*_i$'s are planar.

Now we construct the MSO$_2$ formula corresponding to
Properties~\ref{property0-1}--\ref{property3}.  The formula depends on
the drawing $D_\fedges$ of the set of frame edges~$\fedges$.

\begin{align*}
  \textsc{realizable}_{D_\fedges}(G) ~\equiv~
  & (\exists e_1, \ldots, e_\beta, E_0, E_1,\dots,E_\gamma, v_1,v_2,\dots,v_{\xi}, V_0, V_1,\dots,V_\eta )
  \\
  & \Big[E_0=\{e_1, \ldots, e_\beta\} \wedge V_0=\{v_1,v_2,\dots,v_{\xi}\}\\
  & \wedge~ \textsc{Partition}(E; E_0, E_1, \ldots, E_\gamma)\\
  & \wedge~ \textsc{Partition}(V; V_0, V_1 \ldots, V_\eta)\\
  & \wedge~ \textsc{Incident}(V_0, E_0)\\
  & \wedge~ \theta_3(V_0, E_0)~ \wedge~\theta_4(V_0)~ \wedge~\theta_5(V_0,E_0) \\
  & \wedge~ \big(\forall j \in \{1, 2, \dots, \eta\}\big)\big(\forall v \in V_j\big)
  \big(\forall e \in E\big)\big[I(e, v) \Rightarrow e \in E_{i_j}\big]\\
  &\wedge~\iota(e_1, \ldots, e_\beta,E_1,\dots,E_\gamma, v_1,v_2,\dots,v_{\xi}, V_1,\dots,V_\eta )
\Big].
\end{align*}

\end{document}